\documentclass[conference]{IEEEtran}
\usepackage{amsmath}
\usepackage{amssymb}
\usepackage{graphicx}% Include figure files
\usepackage{epsfig}% Include figure files
\usepackage{bm}% bold math
\usepackage{color}
\newcommand{\mse}{\mathsf{mse}}
\newcommand{\MMSE}{\mathsf{mmse}}

\ifCLASSINFOpdf
\else
\fi
\begin{document}
\title{Signal recovery using expectation consistent approximation for linear observations}

\author{\IEEEauthorblockN{Yoshiyuki Kabashima }
\IEEEauthorblockA{Dept. of Comput. Intell. \& Syst. Sci.\\
Tokyo Institute of Technology\\
Yokohama 226--8502, Japan\\
Email: kaba@dis.titech.ac.jp}
\and
\IEEEauthorblockN{Mikko Vehkaper\"a}
\IEEEauthorblockA{
Department of Signal Processing and Acoustics\\
Aalto University\\
Otakaari 5 A FI-02150 Espoo, Finland \\
Email: mikko.vehkapera@aalto.fi}}
%\and
%\IEEEauthorblockN{James Kirk\\ and Montgomery Scott}
%\IEEEauthorblockA{Starfleet Academy\\
%San Francisco, California 96678-2391\\


%Telephone: (800) 555--1212, Fax: (888) 555--1212}


\newcommand{\bF}{{\textbf {F}}}
\newcommand{\bG}{{\textbf {G}}}
\newcommand{\bW}{{\textbf {W}}}
\newcommand{\bh}{{\textbf {h}}}
\renewcommand{\bm}{{\textbf {m}}}
\newcommand{\bx}{{\textbf {x}}}
\newcommand{\bo}{{\textbf {o}}}
\newcommand{\bs}{{\textbf {s}}}
\newcommand{\bv}{{\textbf {v}}}
\newcommand{\by}{{\textbf {y}}}
\newcommand{\bn}{{\textbf {n}}}
\newcommand{\be}{\begin{equation}}
\newcommand{\ee}{\end{equation}}
\newcommand{\bea}{\begin{eqnarray}}
\newcommand{\eea}{\end{eqnarray}}
\newcommand{\mc}{\mathcal}
\newcommand{\eps}{\varepsilon}
\newcommand{\s}{\sigma}
\renewcommand{\d}{\text {d}}
\newcommand{\e}{\text {e}}
\newcommand{\ds}{\Delta\sigma}
\newcommand{\hw}{h^{\rm W}}
\newcommand{\sw}{\sigma^{\rm W}}
\newcommand{\de}{\partial}
\newcommand{\eq}{\text{ eq}}
\newcommand{\sign}{\text{ sign}}
\newcommand{\ua}{\uparrow}
\newcommand{\da}{\downarrow}
\newcommand{\erfc}{{\rm erfc}}
\newcommand{\figwidth}{7.2cm}

\newtheorem{theorem}{Theorem}

\maketitle

\begin{abstract}
%\boldmath
%In this paper, 
A signal recovery scheme is
developed for linear observation systems based on
expectation consistent (EC) mean field approximation.
Approximate message passing (AMP) is known to be
consistent with the results obtained using the replica theory,
which is supposed to be exact in the large system limit,
when each entry of the observation matrix is independently generated
from an identical distribution.
However, this is not necessarily the case for general matrices.
We show that EC recovery exhibits consistency
with the replica theory for a wider class of random observation matrices.
This is numerically confirmed by experiments for the Bayesian optimal
signal recovery of compressed sensing using random row-orthogonal matrices.
\end{abstract}


% \fi
\IEEEpeerreviewmaketitle

\section{Introduction}
%In many problems of signal processing, observations are assumed to be of the form
%of linear transformation.
Let us suppose that an original $N$-dimensional vector
$\bx =(x_i)\in \mathbb{R}^N$ is transformed into an $M$-dimensional vector $\by =(y_\mu) \in \mathbb{R}^M$ by
a %measurement
matrix $A=(A_{\mu i}) \in \mathbb{R}^{M \times N}$ as
\begin{eqnarray}
\by=A\bx +\bn,
\label{observation}
\end{eqnarray}
where $\bn=(n_\mu)\in \mathbb{R}^M$ is
provided randomly.
Many problems of signal processing are formulated using this form.
Equation (\ref{observation}) describes the basic signal sampling scheme
if we identify $\by$, $\bx$, and $A$ as sampled signal values, Fourier coefficients,
and a Fourier matrix, respectively. In code division multiple access (CDMA) systems, $\bx$ and $A$ correspond to
transmitted signals of $N$ users and a set of signature sequences, whereas
$\by$ is the signal observed at a base station.
In multi-input and multi-output communication (MIMO) systems,
$\bx$, $\by$, and $A$ represent the
transmitted and received signals by $N$ and $M$ antennas, and
the signal transmission efficiency between the input and output antennas. %, respectively.
In compressive sensing (CS), $\bx$, $\by$, and $A$ represent a sparse signal, its measurement,
and a measurement matrix.

For simplicity, we hereafter assume that $A$ is known exactly.
Then, a major problem is to design a computationally efficient scheme for recovering
$\bx$ from $\by$ accurately.
A standard approach for this is to follow the least-square principle;
minimizing $||\by-A\bx||^2$ in conjunction with appropriate $l_2$-regularization terms with respect to $\bx$ yields a signal recovery scheme that 
performs with a low computational cost using operations of linear algebra.
Unfortunately, the optimality of inference accuracy is not guaranteed for
the resulting scheme unless $\bx$ follows a distribution of a specific class.
In recent years, significant attention has been paid to the usage of the $l_1$-norm regularization
when $\bx$ is supposed to be a sparse signal.  %\cite{LASSO}.
The $l_1$-based recovery is capable of recovering sparse signals
with a computational cost of the polynomial order of $N$.
However, this still does not achieve the optimal accuracy in general
\cite{RanganFletcherGoyal}, whereas perfect recovery is possible
for the noiseless case if the observation ratio $\alpha=M/N$ is sufficiently large
\cite{DonohoTanner, KabashimaWadayamaTanaka, DonohoMontanari}.

When the prior distribution of $\bx$ and the distribution of the observation noise
$\bn$ are known, the Bayesian framework offers an optimal recovery scheme %\cite{Iba}
{in minimum mean square error (MMSE) sense}
although its exact execution is computationally difficult in most cases.
The purpose of this paper is to develop a computationally feasible
approximate scheme for the Bayesian signal recovery
for a class of random observation matrix $A$.
For this, we employ an advanced mean field method
known as {\em expectation consistent} (EC) approximation developed in
statistical mechanics \cite{OpperWinther2001,OpperWinther2006} and machine learning \cite{Minka}.
The developed scheme exhibits consistency
with the replica theory, which is supposed to provide exact predictions in the large system limit.

\section{Related work}

{%
Reference \cite{Guo-Baron-Shamai-2009} used the
replica method to find a decoupled formulation for the
input-output statistics of a CS system whose measurement
matrix is composed of independently and identically distributed (i.i.d.) entries.  As a corollary,
this leads to a computationally feasible characterization
of the MMSE as well.  The MMSE of a similar i.i.d. setup
was later evaluated directly in \cite{Huleihel-Merhav-arxiv2013}
by using mathematically rigorous methods.  Numerical
results therein verified the accuracy of the earlier
replica analysis. Finally, non-i.i.d. sensing matrices
where considered in \cite{Tulino-etal-trit2013}, where the
replica method was used to find the support recovery
performance of a class of CS systems.
}

To the best of our knowledge, computationally feasible algorithms approximately
performing the Bayesian recovery were initially developed for
a simple perceptron (linear classifier) \cite{Opper1996}
and later for CDMA \cite{Kabashima,TanakaOkada}.
Recently, a similar idea was applied for CS
\cite{DonohoMontanari, Schneiter, Krzakala} as {\em approximate message passing} (AMP), and was summarized
as a general formulation termed %as 
{\em generalized approximate message passing}
(GAMP) \cite{Rangan}.
However, these studies rely on the assumption that each entry of
$A$, $A_{ij}$, is %generated independently of one another, 
i.i.d., and
the appropriateness for the employment to other ensembles is not guaranteed.
In fact, the necessity for considering a certain characteristic
feature of $A$ in constructing the approximation was pointed out
in \cite{OpperWinther2001}, and its significance was tested for
the simple perceptron \cite{ShinzatoKabashimaB},
CDMA \cite{TakedaUdaKabashima}, and MIMO \cite{HatabuTakedaKabashima}.
Here, we show how this approach is employed for the signal recovery of 
linear observations and examine its significance for an example of CS.

\section{Problem setup}
\subsection{Model specification}
In the following, we suppose that each entry of $\bx$, $x_i$ $(i=1,2,\ldots,N)$,
is generated from a distribution $P(x)$ independently of one another.
For simplicity, we focus on the case where the observation noise $\bn$ obeys a memoryless
zero mean Gaussian distribution, so that the conditional distribution of $\by$ given
$\bx$ is provided as
\begin{eqnarray}
P(\by|\bx, A)=\frac{1}{(2\pi \sigma^2)^{M/2}}
\exp\left ({-\frac{||\by-A\bx||^2}{2\sigma^2}} \right ).
\label{Gaussian}
\end{eqnarray}
General treatment that includes the case of non-Gaussian noise can be found in \cite{Kabashima2008}.
Further, we assume that for eigenvalue decomposition
$A^{\rm T} A= O D O^{\rm T}$, where $O$ is the right eigenbasis of $A$
and $D =(d_i \delta_{ij})$ is the diagonal matrix composed of
eigenvalues $d_i$ of $A^{\rm T}A$,
$O$ can be regarded as a random sample from the uniform distribution of $N \times N$ orthogonal matrices
and $\rho_{A^{\rm T} A}(\lambda)=N^{-1}\sum_{i=1}^N \delta(\lambda-d_i)$
asymptotically converges to a certain distribution $\rho(\lambda)$ with a probability of unity as $N \to \infty$.
This assumption holds when $A_{ij}$'s are generated independently of one another
from a zero mean Gaussian distribution. Further, %there 
this is also the case when $A$ is constructed by
randomly selecting $M$ rows from a randomly generated $N \times N$ orthogonal matrix.

\subsection{Bayesian recovery and expected performance}
Let $\hat{\bx}(\by)$ be an arbitrary recovery scheme  given $\by$. 
Under the above assumption, the mean square error 
${\mse}=N^{-1} \int d\bx d\by P(\bx,\by|A)||\bx -\hat{\bx}(\by) ||^2$
is minimized by the Bayesian recovery
\begin{eqnarray}
\hat{\bx}^{\rm Bayes}(\by) \equiv
\int d \bx \bx P(\bx|\by,A), 
\label{Bayes_recovery}
\end{eqnarray}
which achieves the minimum value of $\mse$ (MMSE)  as
\begin{eqnarray}
{\MMSE}=N^{-1}
\left ( \left \langle \left |\left |\bx \right |\right |^2 \right \rangle -
 \int d\by P(\by|A) \left |\left| \left \langle \bx \right \rangle_{|\by} \right |\right|^2 \right ), 
\label{bayestheorem}
\end{eqnarray} 
where $P(\bx,\by|A)=P(\by|\bx,A) \prod_{i=1}^N P(x_i)$, 
$P(\by|A)=\int d \bx P(\bx,\by|A)$. 
$\left \langle \cdots \right \rangle $ and 
$\left \langle \cdots \right \rangle_{|\by}$ denote
averages with respect to the prior and posterior 
distributions
$P(\bx)=\prod_{i=1}^NP(x_i)$ and 
$P(\bx|\by,A)=P(\bx,\by|A)/P(\by|A)$, 
respectively. 
% Equation (\ref{Bayes_recovery}) is often termed MMSE estimator. 

%Mean square error %\footnote{A similar argument  generally holds for another performance measure.} 
%${\rm mse}=N^{-1} \int d\bx d\by P(\bx,\by|A)
%|\bx -\hat{\bx}(\by) |^2$ is bounded from below as
%\begin{eqnarray}
%{\rm mse} \ge N^{-1}
%\left ( \left \langle \left |\bx \right |^2 \right \rangle -
% \int d\by P(\by|A) \left | \left \langle \bx \right \rangle_{|\by} \right |^2 \right ), 
%\label{bayestheorem}
%\end{eqnarray} 
%where $P(\bx,\by|A)=P(\by|\bx,A) \prod_{i=1}^N P(x_i)$
%and $P(\by|A)=\int d \bx P(\bx,\by|A)$. 
%$\left \langle \cdots \right \rangle $ and 
%$\left \langle \cdots \right \rangle_{|\by}$ denote
%averages with respect to the prior and posterior 
%distributions
%$P(\bx)=\prod_{i=1}^NP(x_i)$ and 
%$P(\bx|\by,A)=P(\bx,\by|A)/P(\by|A)$, 
%respectively. 
%The equality is achieved by the Bayesian recovery, 
%which accords to the minimum mean square estimator
%in the current case
%\begin{eqnarray}
%\hat{\bx}(\by)=\hat{\bx}^{\rm Bayes}(\by) \equiv\left \langle \bx \right \rangle_{|\by}. 
%\label{Bayes_recovery}
%\end{eqnarray}
%\end{theorem}
% 
%\noindent {\bf Proof:}
%Use the Bayes formula 
%$P(\bx,\by|A)=P(\bx|\by,A)P(\by|A)$
%and minimize the integrand of the right hand side of ${\rm mse}$ 
%with respect to $\hat{\bx}(\by)$ for each $\by$. 
%\hfill $\Box$\\
  
%\subsection{Expected performance}
Although optimality of (\ref{Bayes_recovery}) is guaranteed,  
evaluating {the MMSE} is generally difficult. The replica method from statistical 
mechanics enables the evaluation for the large system limit 
$N,M \to \infty$ keeping $\alpha =N/M \sim O(1)$
although its mathematical validity is still open. 
For generality, let us suppose that the true prior and the variance of Gaussian noise, 
$P_0\!(\!x\!)$ and $\sigma_0^2$, may be different from 
$P(\!x\!)$ and $\sigma^2$, respectively. 
The replica symmetric (RS) computation \cite{Beyond,Nishimori} 
evaluates the performance of the Bayesian recovery as follows. 
\begin{theorem}
%The typical value of ${\mse}$ is evaluated as
The RS evaluation offers the typical value of ${\mse}$ as 
\begin{eqnarray}
\label{theorem2}
{\mse} =q-2m+Q_0. 
\end{eqnarray}
Here, $Q_0=\int dx x^2 P_0(x)$, and $q$ and $m$ are determined by extremizing 
the variational free energy density
\begin{eqnarray}
&&\phi(q,m,Q,\hat{q},\hat{m},\hat{Q})=-\frac{\hat{Q}Q}{2}\!-\!\frac{\hat{q}q}{2}\!+\!\hat{m}m 
\!- \!G\left ( \!-\frac{Q\!-\!q}{\sigma^2}\! \right )\cr
&&+\left (\frac{q-2m+Q_0}{\sigma^2}-\frac{\sigma_0^2(Q-q)}{\sigma^4} \right )G^\prime \left (-\frac{Q-q}{\sigma^2} \right ) \cr
&& \! - \! \int \! \! dx^0 \! P_0 \!( x^0 \!)  {\rm D}z \!
\ln \! \left [\!  \int \! dx P(x)\e^{- \!\frac{\hat{Q}\!+\!\hat{q}}{2}\!x^2\!+(\!\sqrt{\hat{q}}z\!+\!\hat{m}x^0\!) x} \!\right ], 
\label{replica_free_entropy}
\end{eqnarray}
where 
${\rm D}z=\frac{dz}{\sqrt{2\pi}}\e^{-\frac{z^2}{2}}$ stands for the Gaussian measure. 
$G(x)=\mathop{\rm extr}_{\Lambda} \left \{
-\frac{1}{2}\int d\lambda \rho(\lambda) \ln \left |\Lambda -\lambda \right | +\frac{1}{2}\Lambda x \right \}
-\frac{1}{2} \ln |x| -\frac{1}{2}$,
{where $\mathop{\rm extr}_{X}\left \{ \cdots \right \}$ denotes the extremization with respect to $X$,}
means the asymptotic form of the single rank Harish-Chandra-Itykson-Zuber integral of $A^{\rm T} A$ \cite{Tanaka2007}, 
which is linked to the ${R}$-transform \cite{MullerGuoMoustakas} as ${R}_{A^{\rm T}A}(x)=G^\prime(x)$. 
\end{theorem}
\begin{proof}
Use techniques employed in \cite{TakedaUdaKabashima,VehkaperaKabashima}\footnote{In \cite{VehkaperaKabashima}, the free energy is expressed using the Stieltjes transform. The two expressions are, however, mathematically equivalent, and always transformable to each other}.
%\hfill $\Box$
\end{proof}

When the correct prior and variance, $P(x)=P_0(x)$ and $\sigma^2=\sigma_0^2$, are
used, the replica symmetry ensures that the dominant solution extremizing
(\ref{replica_free_entropy}) satisfies $Q=Q_0$, $q=m$, $\hat{Q}=0$, and $\hat{q}=\hat{m}$,
%which yields ${\MMSE}=2(Q_0-q)$. 
which yields ${\MMSE}=Q_0-q$. 
It is strongly conjectured that solutions of this type are
always thermodynamically dominant
offering exact (but not rigorous) predictions in the large system limit \cite{NishimoriSherrington,Krzakala}.
Therefore, our goal is to develop a computationally feasible scheme
that approximately evaluates (\ref{Bayes_recovery}) and becomes
consistent with the results predicted by (\ref{replica_free_entropy})
as the system size tends to infinity.

\section{Expectation consistent signal recovery}
\subsection{Gibbs free energy formalism}
The following theorem 
constitutes the basis of our approximation.
\begin{theorem}
\label{theorem3}
Let us define {\em Gibbs free energy} as
\begin{eqnarray}
\Phi(\bm)=\mathop{\rm extr}_{\bh}\left \{
\bh \cdot \bm -\ln \left [
\int d \bx P(\bx|\by,A)\e^{\bh \cdot \bx}
\right ]
\right \}.
\label{Gibbs}
\end{eqnarray}
%The point globally minimizing $\Phi(\bm)$ corresponds to
%$\left \langle \bx \right \rangle_{|\by}$.
{The global minimizer of $\Phi(\bm)$ is
$\bm=\left \langle \bx \right \rangle_{|\by}$}.
\end{theorem}
%\noindent {\bf Proof:}
\begin{proof}
The extremization of (\ref{Gibbs}) offers
\begin{eqnarray}
m_i=\frac{\int d\bx x_i P(\bx|\by,A)\e^{\bh \cdot \bx}}
{\int d \bx P(\bx|\by,A)\e^{\bh \cdot \bx}}.
\label{firstderivative}
\end{eqnarray}
This means that for a given value of
$\bm$, $\bh$ is determined so that the average of $\bx$
for a modified distribution $P(\bx|\by,A,\bh)=P(\bx|\by,A)\e^{\bh \cdot \bx}/
\int d \bx P(\bx|\by,A)\e^{\bh \cdot \bx}$ coincides with $\bm$.
In particular, $\bh=\textbf{0}$ offers $\bm=\left \langle \bx \right \rangle_{|\by}$
and corresponds to an extremum point of $\Phi(\bm)$ since
$\partial \Phi(\bm)/\partial m_i=h_i=0$ holds. Furthermore, $\bm=\left \langle \bx \right \rangle_{|\by}$ is
characterized as the globally minimum point, which is shown as follows.
For any value of $\bm$, the Hessian of $\Phi(\bm)$ is evaluated as
$\left (\frac{\partial^2 \Phi(\bm)}{\partial m_i \partial m_j} \right )=
\left (\frac{\partial h_j}{\partial m_i} \right )=\left (\frac{\partial m_j}{\partial h_i} \right )^{-1}$.
However, (\ref{firstderivative}) indicates that $\frac{\partial m_j}{\partial h_i}$ coincides with
the covariance of $x_i$ and $x_j$ evaluated by $P(\bx|\by,A,\bh)$.
Therefore, both matrices $\left (\frac{\partial m_j}{\partial h_i} \right )$ and
$\left (\frac{\partial^2 \Phi(\bm)}{\partial m_i \partial m_j} \right )=\left (\frac{\partial m_j}{\partial h_i} \right )^{-1}$
are positive definite. This means that $\Phi(\bm)$ is a convex downward function and has a unique minimum point.
%\hfill $\Box$
\end{proof}

\subsection{Expectation consistent approximation}
Theorem \ref{theorem3} indicates that Bayesian recovery
can be performed using the techniques of convex optimization if $\Phi(\bm)$ is correctly evaluated.
Unfortunately, this is also practically unfeasible in most cases as
the assessment of $\Phi(\bm)$ is computationally difficult in general.
One could exceptionally evaluate $\phi(\bm)$ with a low computational
cost if $P(\bx|\by,A)$ were a factorized distribution as $P(\bx)=\prod_{i=1}^N P(x_i)$.
Reference \cite{Plefka} developed an approximation scheme based on
Taylor's expansion around the factorized distribution
by introducing an expansion parameter $\beta $ in the interaction terms
that result in computational difficulty.
In the current case, this implies that the evaluation of
$\Phi(\bm)=\tilde{\Phi}(\bm;\beta=1)$ is performed such that
$\Phi(\bm) = \tilde{\Phi}(\bm;0)+\frac{\partial }{\partial \beta}\tilde{\Phi}(\bm;0)
+\frac{\partial^2 }{2!  \partial \beta^2}\tilde{\Phi}(\bm;0)+\ldots$
by introducing generalized Gibbs
free energy
\begin{eqnarray}
&&\tilde{\Phi}(\bm; \beta )
=const +\mathop{\rm extr}_{\bh}\left \{
\bh \cdot \bm \right . \cr
&& \hspace*{.5cm} \left .-\ln \left [\int d \bx \e^{-\frac{\beta }{2\sigma^2}
\left |\left |\by-A \bx \right |\right |^2} \prod_{i=1}^N \left (P(x_i)\e^{h_ix_i} \right )
\right ]
\right \}.
\label{Plefka}
\end{eqnarray}

This treatment leads to asymptotically exact results for some systems
as statistical properties of the interaction
matrix allow us to truncate the expansion up to the second order \cite{Plefka}
or enable us to sum up %all Taylor series analytically \cite{Nakanishi}.
all relevant terms in Taylor series analytically \cite{Nakanishi}.
In fact, when $A_{ij}$'s are independently generated from
the zero mean variance $M^{-1}$ Gaussian distribution (i.i.d. Gaussian ensemble),
the expansion yields an expression
\begin{eqnarray}
&&\Phi(\bm ) \simeq \mathop{\rm extr}_{Q,E,\bh} \left \{
\frac{1}{2\sigma^2}||\by-A \bm||^2+\frac{M}{2}\ln \left (1+\frac{Q-q}{\alpha \sigma^2} \right ) \right .\cr
&& \left .
 -\frac{NEQ}{2}\!+\!\bh\cdot \bm
\!- \!\sum_{i=1}^N \ln \! \left [\!\int d x_i P(x_i)\e^{-\frac{E}{2}x_i^2+h_i x_i} \!\right ] \!\right \}\cr
&&+const,
\label{GaussianIID}
\end{eqnarray}
for large $N$ and $M$ owing to the latter property, where $q=N^{-1} ||\bm||^2$.
{Notation of ``$\simeq$'' means that the equation 
holds approximately.}
Under appropriate conditions, its minimum is guaranteed to converge to the fixed point of AMP for large systems \cite{Krzakala},
and the treatment becomes asymptotically exact.

Unfortunately, summing up %the 
all the relevant terms in Taylor series %entirely 
for generic matrices is
technically difficult.
For avoiding this difficulty, we employ an alternative approach based on an identity
$\tilde{\Phi}(\bm; 1 )-\tilde{\Phi}(\bm; 0 )=\int_{0}^1 d \beta \frac{\partial }{\partial \beta} \tilde{\Phi}(\bm; \beta )
=\frac{1}{2 \sigma^2} \int_{0}^1 d \beta \left \langle ||\by-A \bx||^2 \right \rangle_{\beta}$,
following \cite{OpperWinther2001,OpperWinther2006}.
Here, $\left \langle \cdots \right \rangle_\beta$ denotes the average with respect to
the modified distribution $P_\beta(\bx|\by,A) \propto \e^{-\frac{\beta }{2\sigma^2}| |\by-A \bx | |^2} \prod_{i=1}^N \left (P(x_i)\e^{h_ix_i} \right ) $ in which $\bh$ is determined so that $\left \langle \bx \right \rangle_\beta=\bm$ holds for each $\beta$.
As a decomposition $\left \langle ||\by-A \bx||^2 \right \rangle_{\beta}=||\by-A \bm||^2  +
{\rm Tr} \left (A^{\rm T} A C_\beta \right )$ is allowed, where $C_\beta=(\left \langle x_i x_j \right \rangle_\beta - m_i m_j)$, evaluating the second moment $\left \langle x_i x_j \right \rangle_\beta$ is necessary to perform the integral of the last expression.
Here, we approximately perform this by replacing
$P_{\beta }(\bx|\by,A)$ with a Gaussian distribution
$P^{\rm G}_{\beta}(\bx|\by,A) \propto \e^{-\frac{\beta }{2\sigma^2}\left |\left |\by-A \bx \right |\right|^2-\frac{\Lambda}{2} ||\bx||^2+\bh^G \cdot \bx}$,
where $\bh^{\rm G}$ and $\Lambda$
are determined so that
the first moment $\bm = \left \langle \bx \right \rangle_{\beta}$ and
the {macroscopic second moment} $Q=N^{-1} \left \langle ||\bx ||^2 \right \rangle_\beta $
are consistent between $P_{\beta}(\bx|\by,A)$ and $P^{\rm G}_{\beta}(\bx|\by,A)$.
Such an approximation scheme is often termed the {\em expectation consistent (EC)} approximation.
This yields the following theorem.
\begin{theorem}
EC approximation offers
\begin{eqnarray}
&&\Phi(\bm ) \simeq \mathop{\rm extr}_{Q,E,\bh} \left \{
\frac{1}{2\sigma^2}||\by-A \bm||^2-NG\left (-\frac{Q-q}{\sigma^2} \right ) \right .\cr
&& \left .
 -\frac{NEQ}{2}\!+\!\bh\cdot \bm
\!- \!\sum_{i=1}^N \ln \! \left [\!\int d x_i P(x_i)\e^{-\frac{E}{2}x_i^2+h_i x_i} \!\right ] \!\right \}\cr
&&+const,
\label{EC}
\end{eqnarray}
for large systems.
\end{theorem}
%\noindent {\bf Proof:}
\begin{proof}
For considering the consistency of $\bm$ and $Q$, we define the generalized Gibbs free energy
as $\tilde{\Phi}(\bm,Q;\beta)=\mathop{\rm extr}_{\bh,E} \! \! \left \{\!
-\! \frac{NEQ}{2} \!+\! \bh \! \cdot \! \bm \!- \!\ln \! \left [ \! \int \! d \bx P(\bx)
\e^{\! -\frac{\beta}{2 \sigma^2}||\by\!-\!A\bx||^2\! -\! \frac{E}{2}||\bx||^2\!+\!\bh\cdot \bx} \! \right ] \! \right \}$,
and denote its Gaussian approximation as $\tilde{\Phi}^{\rm G}(\bm,Q;\beta)$.
EC approximation offers an expression
$\tilde{\Phi}(\bm,Q;1)\simeq $ $ \tilde{\Phi}(\bm,Q;0)+\tilde{\Phi}^{\rm G}(\bm,Q;1)-\tilde{\Phi}^{\rm G}(\bm,Q;0)$.
Each part on the right-hand side is evaluated as follows:
$\tilde{\Phi}(\bm,Q;0) = \mathop{\rm extr}_{\bh,E} \! \!\left \{-\frac{NEQ}{2}+\bh \cdot \bm
\!- \!\sum_{i=1}^N \!\ln \!\left[\!\int dx_i P(x_i) \e^{-\frac{E}{2}x_i^2+h_i x_i} \!\right ]\! \right \}$.
$\tilde{\Phi}^{\rm G}(\bm,Q;1)\!=\!\frac{1}{2\sigma^2}||\by \!- \!A \bm||^2
+ \mathop{\rm extr}_{\Lambda}\left \{\frac{1}{2}\ln \left |\det \left (\Lambda - A^{\rm T}A  \right ) \right | \right .$
$ \left . + \frac{N \Lambda(Q-q)}{2\sigma^2} \right \}+ const$.
$\tilde{\Phi}^{\rm G}(\bm,Q;0)=-\frac{N}{2} \ln \left (\frac{Q-q}{\sigma^2} \right )-\frac{N}{2} + const$.
For $N,M \gg 1$, one can replace
$\ln \left |\det \left (\Lambda - A^{\rm T}A  \right ) \right |$ with $N \int d \lambda \rho(\lambda) \ln |\Lambda -\lambda|$.
Substituting the three expressions in conjunction with this replacement into
the identity $\Phi(\bm)= \mathop{\rm extr}_{Q}\left \{\tilde{\Phi}(\bm,Q;1) \right \}
\simeq \mathop{\rm extr}_{Q}\left \{\tilde{\Phi}(\bm,Q;0)+\tilde{\Phi}^{\rm G}(\bm,Q;1)
 \right . $
$\left . -\tilde{\Phi}^{\rm G}(\bm,Q;0)
\right \}$ yields (\ref{EC}).
%\hfill $\Box$
\end{proof}

Here, two points are worth noting. First, for the current characterization of $A$ based on
the eigenvalue decomposition $A^{\rm T}A = ODO^{\rm T}$, all statistical features of
$A$ are summarized in $G(x)$, which is defined for the asymptotic eigenvalue distribution
$\rho(\lambda)$, in (\ref{EC}). This means that the functional form to be optimized for
computing the Bayesian recovery varies depending on the employed matrix ensemble.
For instance, $G(x)=-\frac{\alpha }{2}\ln \left(1-\alpha^{-1}x \right )$
should be used for the i.i.d. Gaussian ensemble, which reduces (\ref{EC}) to (\ref{GaussianIID}). However, when $A$ is constructed by
randomly selecting $M$ rows from a randomly generated $N\times N$ orthogonal matrix
(row-orthogonal ensemble), the proper function to be employed is given by $G(x)=\mathop{\rm extr}_{\Lambda}\left \{
-\frac{1-\alpha}{2}\ln \Lambda-\frac{\alpha}{2}\ln |\Lambda-\alpha^{-1} |+\frac{1}{2}\Lambda x\right \}
-\frac{1}{2}\ln |x|-\frac{1}{2}$.
This implies that
the employment of AMP (in general, GAMP), the fixed point of which asymptotically extremizes
(\ref{GaussianIID}), for generic matrix ensembles
may not be a theoretically appropriate treatment
even if it leads to a satisfiable approximation accuracy \cite{Barbier}.
Second, although we imposed the consistency of the second moment in a macroscopic manner, one can construct a more
accurate approximation by achieving the consistency in a component wise manner as
for $Q_i=\left \langle x_i^2 \right \rangle_\beta$. Such an approximation was once tested for CDMA demodulation \cite{Winther};
however, it incurs $O(N^3)$ computational costs and is difficult to use for large systems.

\subsection{Consistency with the replica theory}
Following the argument of  \cite{OpperWinther2001},
one can show that EC approximation
becomes asymptotically consistent with the replica theory
for matrix ensembles of the current characterization.
For this, we denote the function to be extremized in (\ref{EC}) as
$\Phi(\bm,Q,\bh,E)$, and introduce the auxiliary partition function
$Y(Q,E;\beta)=\int d\bh d \bm \e^{-\beta \Phi(\bm,Q,\bh,E)}$.
In the limit $\beta \to \infty$, $Y(Q,E;\beta)$ is dominated by the values of
$\bm$ and $\bh$ for which $\Phi(\bm,Q,\bh,E)$ is stationary, provided the paths of
integration are chosen such that the integral exists. Further, assuming the stationarity
with respect to $Q$ and $E$, we have an expression of free energy density as
$f=N^{-1} \mathop{\rm min}_{\bm}
\left \{\Phi(\bm) \right \}
=N^{-1}\mathop{\rm extr}_{Q,E}\left \{-\lim_{\beta \to \infty} \beta^{-1} {\ln Y(Q,E;\beta)} \right \}$.
Variation with respect to $Q$ offers $E=\frac{2}{\sigma^2} G^\prime\left (-(Q-q)/\sigma^2 \right )$.

For assessing the average of $f$ with respect to $A$, $\bx^0$, and $\bn$,
we employ the replica method using the average under the replica symmetric ansatz, 
which offers
\begin{eqnarray}
&&\frac{1}{N} \ln \left [ \exp \left [-\frac{\beta }{2\sigma^2}\sum_{a=1}^n\left | \left |A(\bx^0 - \bm^a)+\bn \right |\right |^2 \right ] \right ]_{A,\bn}\cr
&& = \! n \!\left (\! -\! \left (\frac{\beta (\overline{q}\!-\!2 m \!+\!Q_0)}{\sigma^2}
\! -\! \frac{\sigma_0^2\!\beta^2\! (q\!-\!\overline{q})}{\sigma^4} \!\right )
G^{\prime}\! \left ( \!-\frac{\beta (q\! -\! \overline{q}  )}{\sigma^2} \!\right )\!
\right . \cr
&& \left . \hspace*{1cm} +  G\left (-\frac{\beta (q-\overline{q} )}{\sigma^2}\right )\right ) +O(n^2),
\end{eqnarray}
where we set $q=N^{-1} ||\bm^a||^2$, $\overline{q}=N^{-1} \bm^a \cdot \bm^b$ $(a \ne b)$,
and $m=N^{-1} \bx^0 \cdot \bm^a$.
It is worth noting that $\overline{q} \to q$ holds for $\beta \to \infty$ and
we can identify $\lim_{\beta \to \infty }\beta (q-\overline{q}) =Q-q \equiv \chi$ by a linear response argument.
For $\beta \to \infty$, the integrations over $m_i^a$ and $h_i^a$ can be performed by
using the saddle-point method.
This yields $m_i^a=0$ and $h_i^a=\sqrt{\hat{q}} z_i+\hat{m}x_i^0$ as the saddle point,
where
\begin{eqnarray}
\hat{q} \!=\! \frac{2}{\sigma^2}\!
\left (\!\frac{q\!-\!2m\!+\!Q_0}{\sigma^2}\!-\!\frac{\sigma_0^2 \chi}{\sigma^4} \!\right )
\! G^{\prime \prime} \! \left ( \! -\frac{\chi}{\sigma^2} \!\right )
\!+\!\frac{2\sigma_0^2}{\sigma^4} \!G^{\prime }
\! \left ( \! -\frac{\chi}{\sigma^2} \! \right ),
\end{eqnarray}
\begin{eqnarray}
\hat{m}\!=\!\frac{2}{\sigma^2} G^\prime\left (-\frac{\chi}{\sigma^2} \right ) =E,
\end{eqnarray}
and $z_i$ is a standard Gaussian random variable.
Combining all these, we find the consistency between EC approximation and
the replica theory as
\begin{eqnarray}
\left [ f \right ]_{A,\bx^0,\bn}
=\mathop{\rm extr}_{q,m,Q,\hat{q},\hat{m},\hat{Q}}
\left \{\phi(q,m,Q,\hat{q},\hat{m},\hat{Q} ) \right \}
\end{eqnarray}
by identifying $E=\hat{Q}+\hat{q}$ in (\ref{replica_free_entropy}).

\section{Experimental validation}
We performed numerical experiments for the signal recovery of compressed sensing using the Bernoulli-Gaussian prior
\begin{eqnarray}
P(x)=(1-\rho) \delta(x)+\rho\frac{\exp\left(-\frac{x^2}{2 \sigma_X^2} \right )}{\sqrt{2 \pi \sigma_X^2}},
\label{GBprior}
\end{eqnarray}
for examining the accuracy of the developed scheme.
In the experiments, we set $\rho=0.1$, $\sigma_X^2=1$, and $\sigma^2=0.01$,
and the correct prior and noise value were used for simulating the Bayesian optimal recovery.
The performance was examined for i) row-orthogonal and i.i.d. Gaussian ensembles. 
In addition to these, iii) random $M$ row selection from discrete cosine transform matrix (random DCT),
which does not follow a rotationally invariant distribution but shares the same eigenvalue distribution with
the row-orthogonal ensemble, was tested for investigating the significance of rotational invariance.

\begin{figure}
\begin{center}
\includegraphics[width=6cm,angle=270]{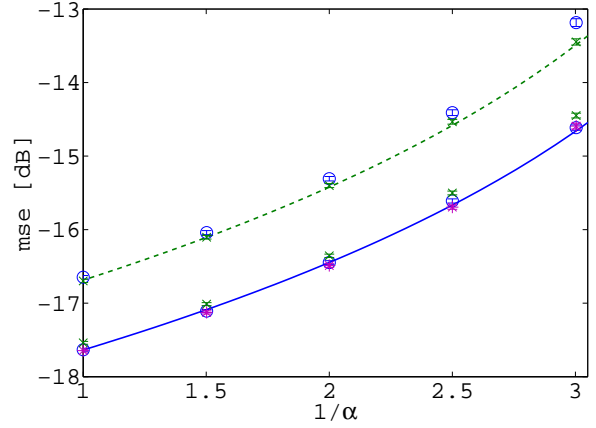}
\end{center}
\caption{
The normalized mean square error $\mathbb{E}\left \{||\bm-\bx^0||^2/||\bx^0||^2 \right \}$  
versus %the compression rate $\alpha$.
$1/\alpha=N/M$.  
Full (blue) and  broken (green) curves represent the theoretical prediction assessed by the replica method
for the row-othogonal and i.i.d. Gaussian ensembles, respectively.
Circle (blue) indicates experimental results obtained by the EC recovery designed for the row-orthogonal ensemble, 
whereas cross (green) stands for those by AMP, which is suitable for the i.i.d. Gaussian ensemble.  
Asterisk (magenta) shows the result for random DCT obtained by the EC recovery designed for
the row-orthogonal ensemble. 
%Normalized mean square error $|\bm-\bx^0|^2/|\bx^0|^2$ achieved by EC approximation.
%Broken (green) and full (blue) curves represent the theoretical prediction assessed by the replica method
%for the i.i.d. Gaussian and row-orthogonal ensembles, respectively.
%Cross (green) and circle (blue) indicate experimental results obtained using EC approximation
%designed for the i.i.d. Gaussian and row-orthogonal ensembles, respectively.
%Asterisk (magenta) represents the result for random DCT obtained using EC approximation designed for
%the row-orthogonal ensemble. 
}
\label{fig1}
\end{figure}

The equation to be solved for the recovery can be read as
\begin{eqnarray}
&&\bh \!=\!\frac{1}{\sigma^2}A^{\rm T}\left (\by-A \bm \right ) + E \bm, \label{recovery_equation0}\\
&&m_i \!=\! \frac{\rho Z(h_i, E) }{1\!-\!\rho\!+\!\rho Z(h_i,E) }\frac{h_i}{E+\sigma_X^{-2}},
\label{recovery_equation}
\end{eqnarray}
\begin{eqnarray}
Q_i \!=\! \frac{\rho Z(h_i, E) }{1\!-\!\rho\!+\! \rho Z(h_i,E) } \!\left (\! \frac{1}{E\!+\!\sigma_X^{-2}}\!+\! \!
\left (\! \frac{h_i}{E\!+\!\sigma_X^{-2} }
\! \right )^{\! 2} \right ) \! ,
\label{recovery_equation2}
\end{eqnarray}
where $Z(h_i,E)=(1+\sigma_X^2E)^{-1/2} \exp \left (\frac{h_i^2}{2(E+\sigma_X^{-2}) } \right )$,
$\chi=N^{-1}\sum_{i=1}^N(Q_i-m_i^2)$, and $E=\frac{2}{\sigma^2} G^\prime\left (-\chi/\sigma^2 \right )$.
The naive iterative substitution scheme did not exhibit a good convergence property.
Therefore, we introduced a dumping factor $\gamma$ and
updated $m_i$ and $\chi$ as $(1-\gamma) m_i +\gamma m_i^{\rm new} \to m_i$
and $(1-\gamma)\chi+\gamma \chi^{\rm new} \to \chi$, where
$m_i^{\rm new}$ and $\chi^{\rm new}$ are the values evaluated from the right-hand sides of
(\ref{recovery_equation}) and (\ref{recovery_equation2}).
For all experiments, we truncated the updates up to $3 \times 10^3$ iterations
setting $\gamma=0.05$, which led to no divergent behavior but
exhibited slower convergence as $\alpha $ decreases.

%Symbols in Fig. \ref{fig1} show the signal recovery performance evaluated from 1000 experiments
%for each condition, whereas curves represent the theoretical prediction assessed by the replica method.
%These indicate the superiority of the row-orthogonal to the i.i.d. Gaussian ensembles in the noisy setting,
%which was also reported for $l_p$-recovery in \cite{VehkaperaKabashima}.
%Excellent agreement between the curves and the symbols for which appropriate functional forms of $G(x)$ were used
%validates the consistency between
%the replica theory and the EC approximation experimentally.
%On the other hand, the slight deviation of symbols for the use of an inappropriate $G(x)$
%indicates the necessity for knowing the statistical properties of the measurement matrix for
%constructing a theoretically proper approximation, whereas its significance becomes smaller as the
%compression rate $\alpha$ increases. The result for random DCT indicates that the performance of
%the row-orthogonal ensembles can be practically achieved with a low computational cost,
%approximately $O(N)$, using the discrete cosine transform similar to
%the noise free case reported in \cite{DonohoTannerUniversal}.
We constructed the EC approximation assuming the row-orthogonal ensemble. 
For comparison, we also tested the performance of AMP
{designed for (\ref{GBprior})}, which is suitable for the i.i.d. Gaussian ensemble.
Symbols in Fig. \ref{fig1} show the signal recovery performance evaluated from $10^3$ experiments
of $N=2^{10}$ systems
while curves stand for the theoretical prediction assessed by the replica method. 
These indicate the superiority of the row-orthonal to the i.i.d. Gaussian ensembles in the noisy setting, 
which was also reported for $l_p$-recovery in \cite{VehkaperaKabashima}. 
Excellent agreement between the circles/crosses and the full/broken curves experimentally validates the consistency between  
the ensemble-dependent proper approximations 
and the replica theory. 
%On the other hand, 
Slight deviation of symbols for the inappropriate recovery schemes 
indicates the necessity for knowing statistical properties of the observation matrix for 
constructing a theoretically proper approximation, whereas its significance becomes smaller as the
compression rate $\alpha$ grows. The result for random DCT indicates that the performance of 
the row-orthogonal ensembles can be practically gained with a low computational cost, 
approximately $O(N)$, %using the discrete cosine transform 
%DCT 
by random row choice of a Fourier matrix 
similarly to the noise free case reported in \cite{DonohoTannerUniversal}.

\section{Summary}
%In summary, 
We developed a computationally feasible approximate scheme of
signal recovery for linear observations affected by Gaussian noises.
The scheme
follows the Gibbs free energy formalism of statistical mechanics and
approximately overcomes the computational difficulty
for evaluating the Gibbs free energy by using a Gaussian approximation
for which the consistency with the true distribution is imposed
for the first moment and a part of the second moment.
The asymptotic consistency with the replica theory is guaranteed for
a class of the measurement matrix ensembles that are characterized by rotational invariance. Experiments for the Bayesian optimal recovery
for compressed sensing using the Bernoulli-Gaussian prior
numerically validated the theoretically obtained results.

The combination of the developed recovery scheme and hyper-parameter estimation \cite{Schneiter,Krzakala}
is under way. Designing a good iteration scheme to solve the recovery equation
(\ref{recovery_equation0})--(\ref{recovery_equation2}) is an interesting and important task.

\section*{Acknowledgments}
The research was funded in part by MEXT KAKENHI Grant No. 25120013 (YK) and
Swedish Research Council under VR Grant 621-2011-1024 (MV).
MV acknowledges the MEXT KAKENHI Grant No. 24106008 for supporting
his visit to the Tokyo Institute of Technology.



%\bibliographystyle{IEEEtran}

\end{document}